\documentclass[journal,a4paper]{IEEEtran}

\usepackage{amsmath}
\usepackage{graphicx}

\newtheorem{theorem}{Theorem}
\newtheorem{lemma}{Lemma}

\newtheorem{proposition}{Proposition}
\newtheorem{remark}{Remark}
\newtheorem{definition}{Definition}
\newtheorem{disc}{Discussion}

\begin{document}

\title{\LARGE A Queueing Characterization of Information Transmission over Block Fading Rayleigh Channels in the Low SNR Regime\footnote{ \hrule\quad\\ This paper has been accepted by \textit{IEEE Trans. on Veh. Technol.} for future publication} }

\author{ Yunquan~Dong,
        Pingyi~Fan,~\IEEEmembership{Senior Member,~IEEE}\\
\{dongyq08@mails,~fpy@mail\}.tsinghua.edu.cn\\
Department of Electronic Engineering, Tsinghua University, Beijing,  China;\\
}

\maketitle

\begin{abstract}

Unlike the AWGN (additive white gaussian noise) channel, fading channels suffer from random channel gains besides the additive Gaussian noise. As a result, the instantaneous channel capacity varies randomly along time, which makes it insufficient to characterize the transmission capability of a fading channel using data rate only. In this paper, the transmission capability of a buffer-aided block Rayleigh fading channel is examined by a constant rate input data stream, and reflected by several parameters such as the average queue length, stationary queue length distribution, packet delay and overflow probability. Both infinite-buffer model and finite-buffer model are considered. Taking advantage of the memoryless property of the service provided by the channel in each block in the the low SNR (signal-to-noise ratio) regime, the information transmission over the channel is formulated as a \textit{discrete time discrete state} $D/G/1$ queueing problem. The obtained results show that block fading channels are unable to support a data rate close to their ergodic capacity, no matter how long the buffer is, even seen from the application layer. For the finite-buffer model, the overflow probability is derived with explicit expression, and is shown to decrease exponentially when buffer size is increased, even when the buffer size is very small.
\end{abstract}


\section{Introduction}
In the pioneering work of Shannon in 1948 \cite{Shannon-MTC}, the capacity of a band limited AWGN (additive white Gaussian noise) channel with an average transmit power constraint was given, which specifies the fundamental transmission limit of AWGN chanels,
\begin{equation}
    C=W\ln\left(1+\frac{P}{N} \right),
\end{equation}
where $W$ is the limited bandwidth and $P/N$ is the SNR (signal-to-noise ratio).

This means that by sufficiently involved encoding systems, although the transmitted signal is perturbed by the random noise, we can still transmit digits at the rate of $C$ {nats} per second for the given SNR, with arbitrarily small frequency of errors.
Later, many capacity approaching coding schemes such as Turbo codes \cite{BGT-turbo} and LDPC codes \cite{Gallerger-LDPC} were developed. In this way, the channels (especially wired ones) can be abstracted into a data pipeline, which delivers data from one place to another at a constant rate, regardless of its contents. Therefore, AWGN channels can be described by the single parameter of data rate. This separability of data rate and other parameters such as delay and overflow probability is also the basis of the OSI (Open Systems Interconnection) model \cite{Z-OSI,W-OSI_tele}, which is widely used in internet/telecommunication systems.

However, due to the fading property of wireless channels, this is very different for wireless communications. In fact, the communication capability of fading channels has also been investigated for more than 70 years, from physical layer to cross-layer ways.

\subsection{Physical layer evaluations of fading channels: Capacity}

Usually, wireless communications suffer from random channel gains besides the additive Gaussian noise. Therefore, the corresponding instantaneous capacity is also a random variable. Furthermore, the wireless channel capacity also depends on what is known about the channel gain $h$ (CSI, channel state information) at the transmitter and receiver. Particularly, the coherent channel model assumes perfect CSI. On the contrary, the noncoherent model not only assumes no CSI available, but also implicitly ignores the correlation between the fading coefficients
over time, since the effect of this correlation will be dwindled by the asymptotic analysis at low SNR (signal-to-noise) \cite{Lizhong-CH_low}. For the coherent channel, transmitter and receiver CSI allows the transmitter to adjust both its power and rate according to the time varying channel gain. In this case, the \textit{ergodic capacity} is defined as
\begin{equation}
    C=I(X;Y,S)=\max\limits_{P(h):\int P(h)p(h)dh=\overline{P}} \textsf{E}_h W\ln\left(1+ \frac{h^2P(h)}{N}\right)
\end{equation}
where $\overline{P}$ is the average power, $P(h)$ is the power allocation policy,  $p(h)$ is the channel gain distribution and $\textsf{E}_h$ is the expectation operator on $h$. According to \cite{GS-Throughput}, the optimal solution is the water-filling power allocation over channel states.

When CSI is only available at the receiver, the transmitter can not perform power adaption. In this case, the transmitted signal is independent from channel gain and the ergodic capacity becomes
\begin{equation}
\begin{split}
    C=&I(X;Y,S)=I(X;S)+I(X;Y|S)
    =I(X;Y|S)\\
    =&\textsf{E}_h W\ln\left(1+ \frac{h^2\overline{P}}{N}\right).
\end{split}
\end{equation}

Besides, its channel capacity can only be given by a multi-character definition  as \cite{SWHL-LSC, DSLS-NCFC}
\begin{equation}
    C=\lim_{n\rightarrow\infty}\frac{1}{n} \sup\limits_{p(x)} I(X^n;Y^n),
\end{equation}
where $p(x)$ is the distribution of the input to the transmitter.

On the one hand, in order to approach ergodic capacity, the code lengths must be long enough to average over the fading process, which is difficult to implement in practice. On the other hand, egodic capacity is the inherent characteristic of the channel and helps understand the performance limits of wireless communications.

However, a transmission strategy for ergodic capacity must incorporate every state of the fading process. Thus those very poor channel states will typically reduce ergodic capacity greatly. To this end, an alternative capacity definition for fading channels with receiver CSI is the \textit{outage capacity}, which specifies the maximum mutual
information rate that can be transmitted in every channel realization except a subset whose probability is less than $\epsilon$ \cite{GS-Throughput, ADTse-OutageClsnr, LA-Outage},
\begin{equation}
    C=\max\limits_{R}\left\{R:~\Pr\left\{W\ln\left(1+\overline{P}/N\right)<R\right\}<\epsilon\right\}.
\end{equation}

It is clear that $\epsilon$-outage capacity considers mainly on the worst cases and on the reliability in the communication. So there are wastes of channel capability in those very good channel conditions.

As the extension of the AWGN channel to fading channels, the concepts as listed above present some descriptions of fading channels using average data rate from the physical layer.

However, wireless networks have been widely deployed for a variety of purposes,
where QoS (quality of service) sensitive applications are increasingly of interest.
Each application has its own specified requirement on throughput, reliability, delay, and delivery ratio. Therefore, only using a long term throughput like ergodic capacity/outage capacity to characterize fading channels is not enough.

In this paper, we consider such a problem of \textit{what kind of service fading channels can provide} and will characterize the transmission capability of a buffer-aided block Rayleigh fading channel from multiple aspects such as data rate, packet delay, stationary queue length distribution and overflow probability.

\subsection{Cross layer evaluations: EC and queueing formulations}
High demands on the QoS, especially when packet delay is the main concern, have fueled substantial research interest in modeling fading channels from upper layers. In \cite{Zorzi-PacketDrop}, the statistics of the packet dropping process due to buffer overflow was investigated. Overflow probability and delay were also studied in \cite{Lee-PacketLoss}, where data arrives in a on-off process and suffers from independent Bernoulli channel errors. The tail distribution of packet delay was studied under both round-robin scheduling and multiuser diversity scheduling algorithms in \cite{Ishizaki-QueingDelay}. Besides, cross-layer joint-adaptation mechanisms
which optimize the physical-layer adaptive modulation and coding and the link-layer
packet fragmentation was proposed in \cite{zhang-JointAMC}. In the paper, more realistic correlated fading channels were discussed using the finite state Markov Chain models (FSMC).

Besides above performance analysis, more fundamental investigations on the transmitting limits of fading channels are also needed from a cross layer perspective. In fact, in face of time varying property of fading channels and the booming variety of traffics in the system, physical layer capacity is no longer a sufficient characterization of the transmission capability of fading channels. On the contrary, the problem what kind of service fading channels can provide can only be answered in a cross layer way.

Particularly, the authors in \cite{DN-EC} proposed a link layer channel model termed effective capacity (EC) as a unified framework to consider parameters such as traffic rate and packet delay bound violations. Using the moment generating method and large deviation theory, two EC functions were established, namely the non-empty buffer probability $\gamma^{(c)}(\mu)$ and decaying exponent $\theta^{(c)}(\mu)$. To summarize, the EC link model aims to characterize wireless channels in terms of functions that can be easily mapped to link-level QoS metrics. However, this channel model involves a lot of inevitable approximation, which restricted its application. 

On the other hand, wireless communications and queueing theory are connected with each other naturally. Firstly, due to the fluctuation of the instantaneous channel capacity, a buffer must be used at the transmitter to match the source traffic with the channel transmission capability in the engineering practice, which is a typical queueing problem. In theory analysis, Gallager has used queueing theory in the collision resolution problem in multi-access channels \cite{Queue-MAC85} in the mid 1980s. Later, in his work with Telatar \cite{Queue-infor95}, they used tools from both queueing theory and information theory to exhibit aspects from collision resolution approach and capacity bounds of a multi-access problom. More recently, Gallager and Berry discussed the optimal power allocation policy under some delay constraints in a finite-buffer aided point-to-point wireless communication system \cite{RG-CoF}. Particularly, two important questions were asked in the conclusion part of the paper. First, due to the fading character of wireless channels, the transmission rate of fading channels depends greatly on the link layer power allocation and rate allocation schemes, which makes the boundary between physical layer and link layer unclear. Therefore, \textit{whether the layered communication protocol design in wired-line communications is reasonable for the wireless situations} is an important problem. For this question, it has been partially answered in \cite{DWF-Deterministic-TWC}, where it is proved that {\em i.i.d.} fading channels can support a constant rate data stream, just like those wired communication links. Therefore, the separation between physical layer and network layer is reasonable in the {i.i.d.} situations. The second problem is \textit{what is the overflow probability of a finite buffer aided wireless link when the arrival rate is constant}. In this paper, this question will be answered for the case of block rayleigh fading channels in the low SNR regime.

In brief, to characterize what kind of service fading channels can provide in a cross-layer way,  the channel can be examined by a constant rate data input to the buffer, and described from multi aspects such as queue length distribution, packet delay, and overflow probability.

However, there are some challenges in applying queueing theory to this problem. For a block fading channel, its channel gain varies block by block. Particularly, its channel gain is a positive real number ranging from zero to infinity. Therefore, the input/output process has to be modeled as a \textit{discrete time continuous state} Markov process, for which few references exist. So one has to resort to other techniques such as stochastic process as in \cite{DF_Overflowrate} or transform the continuous state space to discrete state space by quantization \cite{DF_Discrete}.

\subsection{Overview of this paper}
This paper focuses on characterizing the transmission capability of block fading Rayleigh channels in the low SNR regime, in a cross-layer way.

Although insufficient, previous results in the framework of information theory still serves as an irreplaceable guidance for this work. Therefore, transmission rate is one of the most important parameter for fading channels. In addition, several QoS parameters are evaluated at the same time.

On the other hand, related cross-layer analysis and studies are usually asymptotic results when buffer size is very large. Therefore, this paper focuses on more detailed and explicit characterizations of fading channel transmission capability.

Specifically,  the transmission capability of a fading channels is examined by a constant rate input data stream, and represented in terms of packet delay, queue length distribution and overflow probability.

It is assumed that data arrives at the buffer in packets. Although it is well known that it can be modeled as a Markov process and the overflow probability decreases exponentially, the following challenges exist:

\begin{enumerate}
  \item [1)] The buffer-aided transmission over fading channels corresponds to a \textit{discrete-time continuous-state} Markov process, since the channel gain is continuously distributed while the transmission over the time is in blocks, on which little is known.
  \item [2)] Known results on the overflow probability are asymptotic ones by assuming the buffer size to be very large, explicit expressions are non-available either.
\end{enumerate}

However, in this paper, the data transmission over a block fading channel is formulated successfully as a \textit{discrete-time and discrete-state} queueing problem (Markov Chain), which applies to both small and large buffers. The key idea behind this formulation is the \textit{memoryless property of exponential distributions}. In fact, the service provided by a Rayleigh fading channel in one block is negative exponentially distributed in the low SNR regime. In this situation, although part of the service capability of the block has been consumed by a previous packet (or a part of the packet), the remaining service capability of this block still follows the same distribution as itself, which can be seen as the service provided by a brand new block. In this way, the service time of a packet can be seen as the integer part of its actual service time. Thus, the state space, i.e., the queue length at discrete epochs (the beginning of each block) will be discrete. That is, the original discrete time continuous state queueing process is transformed in to a more trackable discrete time discrete state Markov chain.

With respect to the existing literature, the main contributions
of this paper can be summarized as follows.

\begin{enumerate}
   \item [\textit{a)}] A trackable Markov Chain was established to model the information transmission over block Rayleigh fading channels in the low SNR regime.

   \item [\textit{b)}] For the infinite-buffer model, the stationary distributions was obtained. Closed form of average packet delay was presented, as a function of input data rate.

   \item [\textit{c)}] For the finite-buffer model, the stationary queue length distribution and the average packet delay were derived.

   \item [\textit{d)}] For a given buffer size, the overflow probability is given in an explicit expression, no matter the buffer size is small or large.
\end{enumerate}

The rest of this paper is organized as follows. The communication system model is presented in Section \ref{sec:2} and the corresponding queueing model is formulated in Section \ref{sec:3}. After that, the infinite-buffer model is studied, where the stationary queue length distribution and packet delay is obtained in Section \ref{sec:4}. In Section \ref{sec:5}, the finite-buffer model is investigated and its stationary queue length distribution, packet delay and overflow probability are derived on the basis of results on infinite-buffer model. The obtained result will also be presented via numerical results in Section \ref{sec:6}.  Finally, we will conclude our work in section \ref{sec:7}.

\section{System Model}\label{sec:2}
Consider a point-to-point communication over a block-fading
Rayleigh channel with additive white Gaussian noise \cite{OSW-Itc}. In such channels, the channel gain is modeled by a Rayleigh distributed random variable, which is stationary and ergodic. The channel gain stays fixed over each block of a certain amount of channel uses, and varies independently in different blocks. Let $T_B$ be the block length, $h_n$ be the time varying channel gain during the {\em n}-th block and $\gamma_n=h^2_n$ be the corresponding power gain. Then $h_n$ is Rayleigh distributed and $\gamma_n$ is exponentially distributed, whose probability density functions (\textit{p.d.f.}) are given by, respectively,
\begin{equation}\label{eq:rayle}
    \begin{split}
        p_h(x)=\frac{1}{\sigma^2}e^{\frac{-x^2}{2\sigma^2}} \quad\textrm{and}\quad
        p_\gamma(x)=\frac{1}{2\sigma^2}e^{\frac{-x}{2\sigma^2}},
    \end{split}
\end{equation}
for $x>0$, where $h_n$ is independent of any other $h_{n'}$ for $n\neq n'$,
and so is $\gamma_n$. In fact, $h=\sqrt{r_I^2+r_Q^2}$, where $r_I$
and $r_Q$ are independent Gaussian random variables with zero mean and
variance $\sigma^2$.

The transmission model considered in this paper follows the model given in \cite{RG-CoF}, as shown in Fig.~\ref{fig:TXm}. Assume that the higher layer application maintains a data stream with constant rate $R$. Since the service provided in each block is a random variable, we have to adjust the actual transmission rate over the fading channel
according to the channel states. Therefore, a First In First Out (FIFO) buffer
is used at the transmitter side to match the source traffic stream
with the channel service in each block. Assume that the data are served by packets and the packet size equals to the traffic arriving at the buffer in each block, i.e., $L_p=RT_B$.  Let $Q(n)$ be the queue length of the buffer in \textit{packets} at the start of block $n$. {Note that $Q(n)$ is not necessarily an integer at discrete times of $n^+$. We use $D(n)$ to denote the number of blocks that the packet arriving in block $n$ will spend in the queue.
\begin{figure}[!t]
\centering
\includegraphics[width=7.7cm]{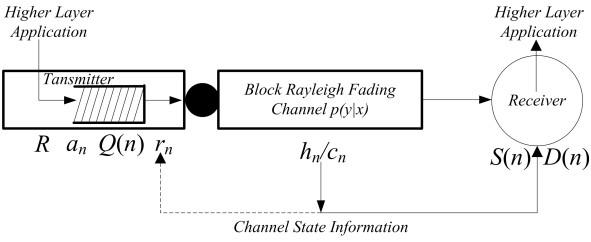}
\caption{The transmission system model} \label{fig:TXm}
\end{figure}

Let $P$ denote the transmit power, $W$ is the transmit
bandwidth and $N_0$ is the noise power spectral density. Let $\alpha$
and $d$ denote the path loss exponent and the distance between the
transmitter and receiver, respectively. Then the instantaneous
capacity of the block Rayleigh fading channel in $nats$ is
\begin{equation}\label{df:insC}
c_n=W\ln\left(1+\frac{\gamma_nPd^{-\alpha}}{WN_0}\right),
\end{equation}
and the service provided by the fading channel in one block is $s_n=c_nT_B$.

Following the concept of channel service in \cite{DWF-Deterministic-TWC}, which combines dimensions of data rate and time, define $S_k$ as the total amount of service provided by the fading channel in $k$ successive blocks, we have
\begin{equation}\label{df:Sk}
S_k=\sum_{m=1}^k s_m.
\end{equation}

The average received SNR is defined as $\rho=\frac{2\sigma^2P}{WN_0d^\alpha}$.

\section{Queueing Model Formulation: Infinite Buffer Size Case}\label{sec:3}
As is shown, the channel service in each block is a random variable of
positive real numbers. Therefore, the queue size $Q(n)$ can not be guaranteed an integer when counting at some fixed epochs such as the beginning of each blocks $n^+$. Therefore, the queueing process is a \textit{discrete time continuous state} Markov process. Unfortunately, few results are available for such processes. However, some transformations are performed on the queueing process in this section, and finally, we will construct a discrete time discrete state $D/G/1$ queue when the average SNR is low. Particularly, the buffer size is assumed to be infinite long.

\subsection{The service time of each packet}
Firstly, the cumulative distribution function (CDF) of the instantaneous capacity $c_n$
is given by
\begin{equation}\label{eq:F_cn_x_orig}
    F_c(x)=1-e^{\frac{-1}{\rho}\left(e^\frac{x}{W}-1\right)}.
\end{equation}

In the low SNR scenario considered in this paper, especially for the wide band communications, we have the following equation holds
\begin{equation}\label{eq:F_cn_x}
    F_c(x)=1-e^{\frac{-x}{W\rho}}.
\end{equation}
\begin{figure}
\centering
\includegraphics[width=3.3in]{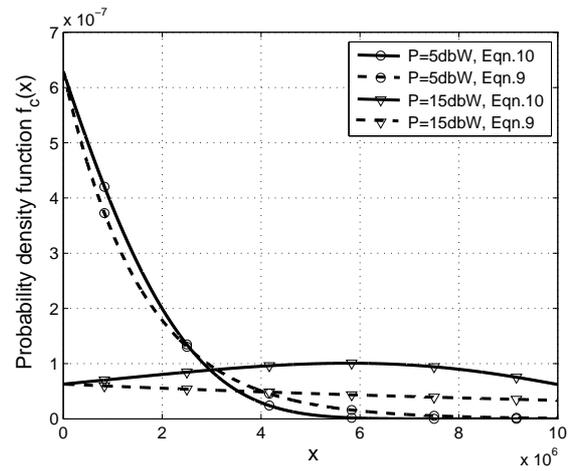}
\caption{The approximation of the {\em p.d.f.} of channel service of each block for middle SNR case}\label{fig:pdf}
\end{figure}

\begin{remark}
It should be noted that, if the average SNR is not so small, or equivalently the transmit power is not very small, the {\em p.d.f.} of channel service in each block is no longer negative exponential distribution and the approximation (\ref{eq:F_cn_x}) of (\ref{eq:F_cn_x_orig}) will not be accurate anymore, as shown by the example in Fig. \ref{fig:pdf}. Therefore, in the case of middle/high SNR or non-Rayleigh fading analysis, the memoryless property of the service provided in each block won't hold any more. As a result, we have to resort to other techniques such as state space quantization or stochastic analysis in \cite{DF_Discrete}.
\end{remark}

Therefore, in the low SNR regime, the CDF of the service in one block ($s_n$) will be given by
\begin{equation}\label{eq:F_sn_x}
    F_s(x)=1-e^{\frac{-x}{\nu}},
\end{equation}
which is an negative exponential distribution, where $\nu$ is defined as
\begin{equation}\label{df:nu}
    \nu=WT_B\rho.
\end{equation}

It can be seen that the service of $k$ successive blocks ($S_k$) follows the Gamma distribution, whose {\em p.d.f.}
is given by
\begin{equation}\label{eq:pdf_Sk}
    f_{S}(x)=\frac{1}{\Gamma(k) \nu^k}x^{k-1}e^{\frac{-x}{\nu}},
\end{equation}
where $\Gamma(k)=\int_0^\infty e^{-t}t^{k-1}dt$ is the Gamma function.

\begin{definition}
    The \textit{service time} $T_n$ of a packet is the number of \textit{complete blocks} of the period in which the service of the packet is finished.
\end{definition}

Therefore, $T_n$ is a non-negative integer random variable. Firstly, the probability that a packet is served within one block can be derived as follows,
\begin{equation}\label{dr:p_0}
    p_0=\Pr\{T_n=0\}=\Pr\{ s_n>L_p \} =e^{\frac{-L_p}{\nu}}\stackrel{(\triangle)}{=}e^{-\theta},
\end{equation}
where $\theta=\frac{L_p}{\nu}$.

\begin{remark}
    We say that the service time of a packet is zero if it is served within one block, for which the detailed explanations are listed as follows. For the block in which the packet is served, although part of its service capability has been consumed, the remaining amount of service follows the same probability distribution as that provided by a complete block, according to the memoryless property of negative exponential distribution. So it can continue to serve the next packet like another complete block. In this sense, the service time of this packet is zero.
\end{remark}

\begin{remark}
    Besides, the ergodic capacity reduces to $C_e=W\rho$ in the low SNR regime. So we have $\theta=\frac{L_p}{\nu}=\frac{tT_B}{WT_B\rho}=\frac{R}{C_e}$. Thus, $\theta$ is the ratio between traffic rate and ergodic capacity. Secondly, from the queueing aspect, $\nu$ is the average amount of service provided in each block. Then $\theta$ is the average number of blocks needed for one packet, which is also the average service completion interval. Finally, $\theta$ can further be interpreted as $\frac{L_p}{\nu}/1$, which is the ratio between the average service time and the average arrival interval, which is often referred as to the \textit{traffic load}.
\end{remark}

Similarly, with $F_s(x)$, $f_{S_k}(x)$ and recall that $S_{k+1}=S_k+s_{k+1}$, the probability that the service time of a packet is $k$ ($k\geq1$) blocks can be obtained as
\begin{equation}\label{dr:p_k}
\begin{split}
    p_k=&\Pr\{T=k\}\stackrel{(a)}{=}\Pr\{S_{k}\leq L_p, S_{k+1}> L_p\}\\
    \stackrel{(b)}{=}&\int_0^{L_p} f_{S_{k}}(x_{k})dx_{k} \int_{L_p-x_{k}}^\infty f_{s}(s_{k+1})ds_{k+1}\\
    =&\frac{1}{k!}e^{\frac{-L_p}{\nu}}\left(\frac{L_p}{\nu}\right)^k
    =\frac{1}{k!}e^{-\theta}\theta^k,
\end{split}
\end{equation}
where (b) holds because $S_k$ and $s_{k+1}$ are independent from each other. Most importantly, (a) is assumed to hold, also in the sense of the memoryless property of negative exponential distributions, which is formally given as follows.

    If $X$ is an negative exponentially distributed random variable, then $X$ is memoryless, namely,
    \begin{equation}
        \Pr\{X\leq s+t|X>s\}=\Pr\{X\leq t\}.
    \end{equation}

Therefore, when we say that the service time of a packet is $T_n=k$, the packet is actually completed in $k+1$ blocks,as shown in (\ref{dr:p_k}.a). However, only a part of the service capability of block $n+1$ is used. So it can continue its service for the next packet in the queue. As is shown by (\ref{eq:F_sn_x}), the channel service of each block in the low SNR regime follows the negative exponential distribution, which is memoryless. Therefore, the remaining service capability of the {\em n}+1-th block can be seen as the service that can be provided by a new block. In this sense, the packet uses only $k$ blocks. Similarly, if a packet is completed within one block as in (\ref{dr:p_0}), it is defined that its service time is zero, i.e., $T_n=0$.

\subsection{The Markov Chain model}

In this way, the problem of information transmission over a block fading Rayleigh channel can be transformed into a classical discrete time $D/G/1$ queueing problem. Its arrival process is $\{A_n=1,n\geq1\}$ (unit: $L_p$) and its service time $T_n$ (unite: block) is a Poisson distributed random variable, whose probability generating function (PGF) is given by
\begin{equation}\label{dr:pgf_t}
\begin{split}
    G(z)=&\textsf{E}[z^{T_n}]
    =e^{-\theta}+\sum_{k=1}^\infty \frac{1}{k!}e^{-\theta} (\theta z)^k
    =e^{\theta(z-1)}.
\end{split}
\end{equation}

According to the relationship of the moments and PGFs, we have
\begin{equation}\label{rt:E_T}
    \textsf{E}[T_n]=G'(t)|_{t=0}=\theta.
\end{equation}

\begin{remark}
    The arrival process considered is constant and one packet arrives in each block. Then the expected number of packets arriving during a service time also equals to $\theta$. Throughout the paper, it is assumed that $\theta<1$ so that the queue is stable.
\end{remark}

Up to now, the information transmission over the channel is formulated as a late arrival system queueing model with immediate access, in which each packet arrives at the end of a block ($k^-$) and leaves at the beginning of the block ($k^+$). The service of each packet is also assumed to start at the beginning of a block. If the service time of a packet is zero, i.e., $T_n=0$, it is assumed to leave the queue immediately at the beginning of the block following its arrival.

Let $Q_n^+$ be the number of packets in the queue at the time of $n^+$. Particularly, $n^+$ may not necessarily a departure epoch of a packet. In this case, $n^+$ is actually lies in the service time of a certain packet. Then the remaining service time of the packet after $n^+$ is a random variable and will determine the value of $Q_{n+1}^+$ together with the value the channel service in block $n+1$. That is to say, the value of $Q_{n+1}^+$ is determined not only by $Q_n^+$ but also by how many of the current packet has been served before $n^+$. Therefore, the arbitrary time queue length process $\{Q_n^+,n\geq0\}$ is usually not a Markov chain.

Define $\tau_n$ be the departure epoch of packet $n$. Assume that the packet will leave the buffer at the same time as packet $n-1$ if its service time is zero, i.e., $T_n=0$ and $\tau_n=\tau_{n-1}$. Therefore, $\tau_n^+$ is the aftereffectless time of the queue length process. Let $L_n^+=Q(\tau_n^+)$ be the number of packets in the buffer after the departure of the {\em n}-th packet. Then $\{L_n^+,n\geq1\}$ is a Markov chain, which is called the embedded Markov chain (EMC) of queue length process $\{Q_n^+,n\geq0\}$. Since the arrival process is $\{A_k=1, k\geq1\}$, i.e., one packet each block, there will be $T_n$ packet arriving in the service time of the $n$-th packet $T_n$.  Then we have
\begin{equation}\label{df:L_n_+}
    L_{n+1}^+=\left\{ \begin{aligned}
             &L_n^+-1+T_n, & L_n^+\geq1\\
             &T_n          & L_n^+=0.
             \end{aligned} \right.
\end{equation}

The transition probability matrix of $\{L_n^+,n\geq1\}$ will be given by
\begin{equation}\label{df:P_matri}
    \textbf{P}=\left[
  \begin{array}{ccccccccccccccccccccc}
    p_0 & p_1 & p_2 & p_2  &\cdots \\
    p_0 & p_1 & p_2 & p_2  &\cdots \\
    0   & p_0 & p_1 & p_2  &\cdots \\
    0   & 0   & p_0 & p_2  &\cdots \\
    \vdots & \vdots & \vdots &\vdots & \ddots \\
  \end{array}
\right]
\end{equation}

The details of packets arrival and departure of some elements ($p_{00},p_{01},p_{10},p_{11}$) in matrix \textbf{P} are illustrated in Fig. \ref{fig_3} as some examples.
\begin{figure}
\centering
\includegraphics[width=3in]{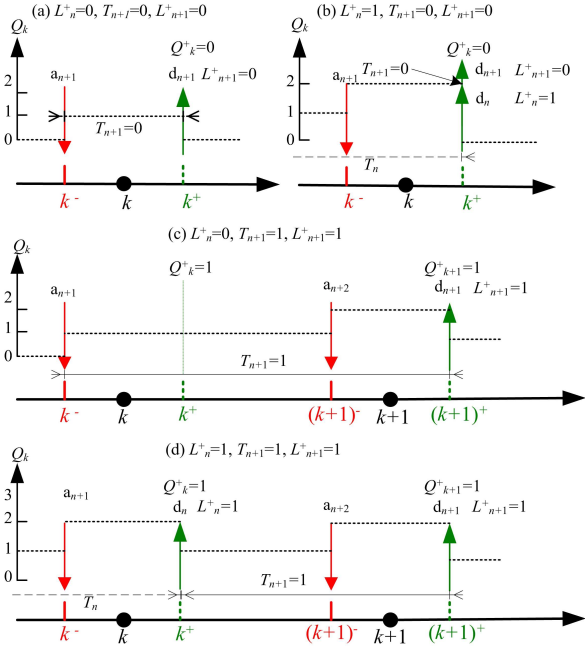}
\caption{Examples for the queueing model}\label{fig_3}
\end{figure}

\begin{disc}
    Poisson Arrivals.

    Suppose the arrival process is a Poisson process with intensity $\lambda$ rather than constant process. Then the number of packets arriving at the buffer in each block will be a random variable of non-negative integer, i.e., $A_k^\lambda\in (-\infty, 0)$. Its PGF is given by
    \begin{equation}
        G_A^\lambda(z)=\textsf{E}[z^{A_k}]=e^{\lambda(z-1)}.
    \end{equation}

    Assume the buffer is empty at the departure of $n$-th packet, the probability that the buffer is still empty at the departure of the $n+1$-th packet will be
    \begin{equation}
        q_0=\sum_{k=0}^\infty \Pr\{ T_n=k \}
                              \Pr\{ \sum_{j=1}^k A_j^\lambda = 0 \}
           =e^{\theta(e^{-\lambda}-1)}.
    \end{equation}

    Similarly, the probability that the queue size is $i$ at the departure of the $n+1$-th packet will be
    \begin{equation}
        q_i=\sum_{k=1}^\infty \Pr\{ T_n=k \}
                              \Pr\{ \sum_{j=1}^k A_j^\lambda = i \}.
    \end{equation}

    Denote $A_{T_n}^\lambda$ be the number of packets arriving in $T_n$ blocks, we know that it is a compound Poisson process, whose PGF is given by
    \begin{equation}
        G_{A_{T_n}}^\lambda=e^{\theta(G_A^\lambda-1)}=e^{\theta \left( e^{\lambda(z-1)} -1\right)}
    \end{equation}

    Therefore, one has
    \begin{equation}
        q_i=\textsf{E}_{T_n}\left[\Pr\{A_{T_n}=i\}\right]=\frac{{G_{A_{T_n}}^\lambda}^{(i)}(0)}{i!}.
    \end{equation}

    In this way, the transition probability matrix of the queue size at departure epochs when arrival process is Poisson is given by:
\begin{equation}\label{df:P_matri}
    \textbf{P}^\lambda=\left[
  \begin{array}{ccccccccccccccccccccc}
    q_0 & q_1 & q_2 & q_2  &\cdots \\
    q_0 & q_1 & q_2 & q_2  &\cdots \\
    0   & q_0 & q_1 & q_2  &\cdots \\
    0   & 0   & q_0 & q_2  &\cdots \\
    \vdots & \vdots & \vdots &\vdots & \ddots \\
  \end{array}
\right]
\end{equation}

The following analysis is focused on the case when the arrival process is constant. Results for Poisson arrival can be obtained by simply replacing $\textbf{P}$ with $\textbf{P}^\lambda$ and are omitted here.
\end{disc}

\section{Queue Length Distribution and Delay Analysis}\label{sec:4}
In this section, we will perform a detailed investigation on the queue length process in terms of stationary queue length distribution and packet delay.

\subsection{The Stationary Queue Length at the Departure Epochs}

Let $L^+=\lim_{n\rightarrow\infty}L_n^+$ be the limit of the queue length process, which is called the \textit{stationary queue length at departure epochs}. Denote
\begin{equation}\label{df:pi_j}
    \pi_j=\Pr\{L^+=j\}=\lim_{n\rightarrow\infty}\Pr\{L_n^+=j\},~~~j\geq0,
\end{equation}
then the vector $\vec{\pi}=\{\pi_0,\pi_1,\cdots\}$ is the stationary queue length at the departure epochs.
\begin{theorem}\label{th:sta_distri_dep}
    If $\theta<1$, the PGF of the stationary queue length at the departure epochs is given by:
    \begin{equation}\label{rt:sta_distri_dep}
        L^+(z)=\frac{ (1-\theta)(1-z) }{ 1-ze^{\theta(1-z)} },
    \end{equation}
    where $\theta=\frac{L_p}{\nu}$.
\end{theorem}

With $L^+(z)$, the average queue length can be obtained by
    \begin{equation}
        \textsf{E}[L^+]=\lim_{z\rightarrow1}{L^+}'(z)=\frac{\theta(2-\theta)}{2(1-\theta)}.
    \end{equation}

The variance of queue length can also be obtained as
\begin{equation}
\begin{split}
    \textsf{Var}[L^+]=&\lim_{z\rightarrow1}\left[{L^+}''(z)+{L^+(z)}'-({L^+(z)}')^2\right]\\
    =&\frac{12\theta-18\theta^2+10\theta^3-\theta^4}{12(1-\theta)^2}.
\end{split}
\end{equation}

While the average queue length determines the buffer size, the standard deviation $\sigma_{L^+}=\sqrt{\textsf{Var}[L^+]}$ is an indicator of the ``spread" of the queue length distribution. For a smaller $\sigma_L$, the queue length tends to be very close to the mean value. For large standard deviations, the queue length is spread out over a wider range of values.

It is known that there is a one-to-one correspondence  between the probability distribution and its PGF. Therefore, the distribution of the queueing process at departure epochs is totally determined by $L^+$, which will be derived later.

Besides, by a similar discussion in \cite{Tian-disQ} on the relationship between the stationary distribution at departure epochs and arbitrary epochs, we have the following ptoposition. Interested readers can refer \cite{Dong-queueingWCSP} for more details, as well as the proof of \textit{Theorem} \ref{th:sta_distri_dep}.

\begin{proposition}\label{th:sta_distri_arbtr}
    The stationary queue length distribution at departure epochs and arbitrary epochs are the same.
\end{proposition}

Therefore, we will have the stationary distribution of the queue length process at arbitrary epochs by \textit{Theorem} \ref{th:sta_distri_dep} and \textit{Proposition} \ref{th:sta_distri_arbtr}.

\subsection{Length of Busy Period}\label{subsec:busy_period}
The \textit{busy period} $B$ is defined as a time interval during which the channel continuously transmitting packets and an \textit{idle period} $I$ is a time interval when the buffer is empty. Thus the system repeats cycles of busy and idle periods.

\begin{proposition}\label{prop:busy-period}
    The PGF and average value of busy period $B$ are given by, respectively
    \begin{equation}
        \begin{split}
            B(z)=\frac{-1}{\theta z} \textsf{W}(-\theta z e^{-\theta}), \quad \textrm{and} \quad
            \textsf{E}[B]=\frac{\theta}{1-\theta},
        \end{split}
    \end{equation}
    where $\textsf{W}(z)$ is the Lanbert\textit{W} function \cite{Lambert_w}. The average idle period is given by $\textsf{E}[I]=\frac{1}{e^\theta-1}$.
\end{proposition}

\begin{proof}
Note that the length of a busy period and the number of packets served during a busy period do not depend on the order of service. Thus, we can assume the Last Come First Service (LCFS) discipline. Then there is a recursive relationship in the busy period. That is, the busy period is composed of the service time of the first packet in the busy period, and all the periods generated by packets that arrive during the service time of the first packet. Since one packet arrives in each block, the number of packets arriving in the first service period equals to the length of the service period. We have
\begin{equation}
    B=T+B^1+B^2+\cdots+B^T,
\end{equation}
where $\{B^n,n=1,2,\cdots,T\}$ are mutually independent random variables with the same distribution as $B$, which leads to
\begin{equation}
    \textsf{E}[z^B|T=k]=\textsf{E}[z^{k+B^1+\cdots+B^k}]=[zB(z)]^k.
\end{equation}

According to the whole probability formula, we have
\begin{equation}\label{eq:bzdr}
    \begin{split}
        B(z)=&\sum_{k=0}^\infty \Pr\{T=k\}\textsf{E}[z^B|T=k]\\
        =&\sum_{k=0}^\infty \Pr\{T=k\}[zB(z)]^k\\
        =&G\left[zB(z)\right],
    \end{split}
\end{equation}
where $G(z)$ is the PGF of the service time and given in (\ref{dr:pgf_t}).

Solving $B(z)$ from (\ref{eq:bzdr}), we have $B(z)=\frac{-1}{\theta z} \textsf{W}(-\theta z e^{-\theta})$ for $|z|<1$.

Then the expected value of busys period can be derived by
\begin{equation}\label{rt:E_busy}
    \begin{split}
        &\textsf{E}[B]=B'(z)|_{z=1}\\
        =&G'[zB(z)][B(z)+zB'(z)]|_{z=1}\\
        =&G'[1][1+B'(1)]\\
        =&\theta (1+\textsf{E}[B]),
    \end{split}
\end{equation}
and we get $\textsf{E}[B]=\frac{\theta}{1-\theta}$.

During the idle period, the buffer is continuously empty. Therefore, packets arriving in this period are all served within one block, namely,
\begin{equation}
\begin{split}
    \Pr\{I=k\}=&\Pr\{s_1>L_p,\cdots,s_k>L_p,s_{k+1}<L_p\}\\
    =&p_0^k(1-p_0);
    \end{split}
\end{equation}
where $p_0=\Pr\{s_n>L_p\}=e^{-\theta}$. It is clear that the idle period is a geometrically distributed random variable. Then the average idle period is given by
\begin{equation}
    \textsf{E}[I]=\sum_{k=1}^\infty k\Pr\{I=k\}=\frac{1}{e^\theta-1}.
\end{equation}

\end{proof}

Because the average service time of a packet is $\textsf{E}[T]=\theta$, the number of packets served in a busy period can be obtained as $\textsf{E}[\Gamma]=\frac{\textsf{E}[B]}{\textsf{E}[T]}=\frac{1}{1-\theta}$.

It can be seen that $\textsf{E}[B]$ is increasing while $\textsf{E}[I]$ is decreasing with $\theta$. The existence of $\textsf{E}[B]$ and $\textsf{E}[I]$ is a result of the fluctuation of the instantaneous channel capacity, which is the key difficulty of evaluating and using fading channels, and is unavoidable. Particularly, larger $\textsf{E}[B]$ means higher channel usage efficiency but larger packet delay. On the contrary, larger $\textsf{E}[I]$ means wastes of channel and less delay. They can not be optimized at the same time. But some trade-offs can be made according to the specified application scenario.

\subsection{Packet Delay}\label{subsec:delay}

The \textit{packet delay} $D$ is defined as the time interval between the arrival of a packet and its departure. Firstly, the packet delay consists of a service time $T$. Secondly, if the packet arrives seeing a non-empty buffer, it must wait for a \textit{waiting time} $W$ for its service. Finally, for the formulation in this paper, there is another piece of time that the packet spends in the system, i.e., the \textit{vestige time} $V$. In this paper, $T=k~(k=0,1,\cdots)$ means that the service of a packet is not finished until the $k+1$-th block. According to the memoryless property of the negative exponential distribution, although part of the service ability has been consumed, it is considered as a brand new block. However, the packet still has to spend a part of that block in the system, which is called the \textit{vestige time}. Thus, we know that
\begin{equation}\label{df:delay}
    D=T+W+V.
\end{equation}

\begin{theorem}\label{th:delay}
    In the FIFO discipline, the average packet delay is given by
        \begin{equation}\label{rt:E_D}
            \textsf{E}[D]=\frac{1}{2}+\theta+\frac{\theta^2}{2(1-\theta)}+\int_0^1 (x-1) e^{\frac{-\theta}{x}} dx.
        \end{equation}
\end{theorem}

\begin{proof}
    By the formulation of the queueing model, it can be seen that the vestige time $V$ is independent of service time $T$ and waiting time $W$. Therefore, they can be obtained separately. See Appendix \ref{prf:them3} for the details of the proof.
\end{proof}

\textit{Theorem} \ref{th:delay} shows the relationship between average packet delay and traffic load $\theta$ (equivalently the input data rate $R$).
As known to all, ergodic capacity is the expected transmission rate averaged over all fading states, which is considered as the limiting transmission rate of a fading channel when physical layer coding delay is allowed to be large enough. Particularly, it was proved in \cite{DWF-Deterministic-TWC} that, \textit{i.i.d.} fading channels can actually support a constant data stream at the rate of ergodic capacity without queueing delay.
However, it is seen from \textit{Theorem} \ref{th:delay} that as $\theta$ approaching 1, the average packet delay will go to infinity. This means that the block fading channel has difficulty in supporting a traffic rate near its ergodic capacity, even when the buffer size is infinite and considered from the application layer, and neglecting its physical layer coding delay.

\section{Queueing Model Formulation: Finite-Buffer Case}\label{sec:5}
In the practical engineering, the size of a buffer must be finite. Thus it is more useful to investigate the finite buffer-aided communications.

Denote the buffer size as $K$. If a packet arrives at the buffer when the queue length is $K$, an overflow happens. In this section, we are interested in the overflow probability. In this case, the ${(K+1)\times(K+1)}$ transition probability matrix of queue length process at the departure time $\{\widehat{L}_n^+,n\geq1\}$ is given by
\begin{equation}\label{df:P_matri finite}
    \widehat{\textbf{P}}=\left[
  \begin{array}{ccccccccccccccccccccc}
    p_0 & p_1 & p_2 & p_3  &\cdots & p_{k-1} & \widehat{p}_K\\
    p_0 & p_1 & p_2 & p_3  &\cdots & p_{k-1} & \widehat{p}_K\\
    0   & p_0 & p_1 & p_2  &\cdots & p_{k-2} & \widehat{p}_{K-1}\\
    0   & 0   & p_0 & p_1  &\cdots & p_{k-3} & \widehat{p}_{K-2}\\
    \vdots & \vdots & \vdots &\ddots & \ddots & \ddots & \vdots\\
    0   & 0   & 0   & 0    &\cdots & p_{1} & \widehat{p}_{2}\\
    0   & 0   & 0   & 0    &\cdots & p_{0} & \widehat{p}_{1}\\
  \end{array}
\right]
\end{equation}
where $\widehat{p}_k=\Pr\{T_n\geq k\}=1-\sum_{j=0}^k\frac{\theta^j}{j!}e^{-\theta}$ is the probability that the number of arrival packets (equal to the service time) is equal or more than $k$ during a service of current packet.

\subsection{Stationary Queue Length Distribution}
Firstly, the stationary distribution of the queue length process for the infinite-buffer model can be obtained with its PGF (\ref{rt:sta_distri_dep}) in \textit{Theorem} \ref{th:sta_distri_dep}. It is noted that $\pi_0=1-\theta$.

Before the discussion on the overflow probability, let's introduce one lemma that will be used.
\begin{lemma}\label{lem:causys_highorder}
    \textit{Cauchy Integral Formula} (\textit{extended}) \cite{Xijiao-Fubian}\\ Let $C$ be a simple closed positively oriented piecewise smooth curve on a domain $D$, and let the function $f(z)$ be analytic in a neighborhood of $C$ and its interior. Then for every $z_0$ in the interior of $C$ and every natural number $n$, we have that $f^{(n)}(z)$ is $n$-times differentiable at $z_0$ and its derivative is
    \begin{equation}\label{eq:cauchy's_inti}
        f^{(n)}(z_0)=\frac{n!}{2\pi i}\oint_C \frac{f(z)}{(z-z_0)^{n+1}} dz,~n=1,2,\cdots.
    \end{equation}
\end{lemma}

Let $C$ and $C'$ are both circles centered on the origin with their radiuses $r<1$ and $1<r'<\frac{-1}{\theta}W_{-1}(-\theta e^{-\theta})$.

Since $L^+(z)$ is convergent within the unit circle, for $k \geq 1$, the stationary can be expressed by
\begin{equation}\label{dr:pi_infty}
\begin{split}
    \pi_k=&\frac{1}{2\pi i}\oint_{C} L^+(z) \frac{dz}{z^{k+1}}
    =\frac{1-\theta}{2\pi i} \oint_{C} \frac{1-z}{1-ze^{\theta(1-z)}} \frac{dz}{z^{k+1}}.
\end{split}
\end{equation}

It can be seen that function $g(z)=\frac{1-z}{1-ze^{\theta(1-z)}} \frac{1}{z^{k+1}}$ has two singular points within circle $C'$, namely $z=1$ and $z=\frac{-1}{\theta}\textsf{W}_{-1}(-\theta e^{-\theta})$, where $\textsf{W}_{-1}(z)$ is the lower branch Lambert W function. Particularly, $z=1$ is one movable singularity since $\lim_{z\rightarrow1}g(z)=1$ is finite. Therefore, within circle $C'$, $g(z)$ can be considered as analytic. According to Cauchy-Goursat's theory, the integral of $g(z)$ along any closed curve in $C'$ is a invariable. Then we have

\begin{equation}\label{dr:pi_infty1}
\begin{split}
    &\pi_k=\frac{1-\theta}{2\pi i} \oint_{C'} \frac{1-z}{1-ze^{\theta(1-z)}} \frac{dz}{z^{k+1}}\\
    =&\frac{1-\theta}{2\pi i} \oint_{C'} \frac{(z-1)\frac{e^{\theta(z-1)}}{z}  }{\frac{e^{\theta(z-1)}}{z}-1} \frac{dz}{z^{k+1}}\\
    \stackrel{(a)}{=}&\frac{1-\theta}{2\pi i} \oint_{C'} (z-1)\frac{e^{\theta(z-1)}}{z} \sum_{j=0}^\infty \left(\frac{e^{\theta(z-1)}}{z}\right)^j \frac{dz}{z^{k+1}}\\
    =&\frac{1-\theta}{2\pi i} \sum_{j=1}^\infty \oint_{C'} e^{j\theta (z-1)} \left( \frac{1}{z^{k+j}} - \frac{1}{z^{k+j+1}} \right) dz\\
    \stackrel{(b)}{=}& (1-\theta) \sum_{j=1}^\infty \left[ \frac{1}{(k+j-1)!}\left( e^{j\theta(z-1)} \right)^{(k+j-1)}
     \left.\frac{1}{(k+j)!}\left( e^{j\theta(z-1)} \right)^{(k+j)} \right]\right|_{z=0}\\
    =&(1-\theta) \sum_{j=1}^\infty \left[ \frac{1}{(k+j-1)!}(j\theta)^{k+j-1} - \frac{1}{(k+j)!}(j\theta)^{k+j}\right] e^{-j\theta}
\end{split}
\end{equation}
where (a) follows $\frac{1}{1-z}=\sum_{j=0}^\infty z^j$ for $|z|<1$ and the fact that $\left|\frac{e^{\theta(z-1)}}{z}\right|<1$ on circle $C'$, (b) follows \textit{Lemma} \ref{lem:causys_highorder} and $e^{j\theta(z-1)}$ is analytic.

To insure the the following general item hold for all $k\geq0$, it is defined that $\varphi_{-1}=1$. Besides, define for $k\geq0$
\begin{equation}\label{rt:fai_0}
    \varphi_k=(1-\theta)\sum_{j=1}^\infty \frac{1}{(k+j)!}(j\theta)^{k+j} e^{-j\theta}
\end{equation}
Then we have
\begin{equation}\label{rt:pij}
    \pi_k=\varphi_{k-1}-\varphi_k.
\end{equation}

Specifically, $\varphi_0=\theta$. A simple proof is given as follows. Denote $x=-\theta e^{-\theta}$ for convenience. Remember the Lambert W function can be written as $\textsf{W}_0(x)=\sum_0^{+\infty} \frac{(-n)^{n-1}}{n!}x^n$ and $\textsf{W}_0'(x)=\frac{\textsf{W}_0(x)}{x(1+\textsf{W}_0(x))}$ holds \cite{Lambert_w}. Then we have
\begin{equation}\label{dr:thta1}
\begin{split}
    \varphi_0=&(1-\theta)\sum_{j=1}^\infty \frac{1}{j!}(j\theta)^j e^{-j\theta}
    =(1-\theta)\sum_{j=1}^\infty \frac{(-jx)^j}{j!}\\
    =&-x(1-\theta)\sum_{j=1}^\infty \frac{(-jx)^{j-1}}{j!}\cdot jx^{k-1}
    =-x(1-\theta)\textsf{W}_0'(x)\\
    =&-x(1-\theta)\frac{\textsf{W}_0(x)}{x(1+\textsf{W}_0(x))}
    =\theta e^{-\theta}(1-\theta)\frac{-\theta}{-\theta e^{-\theta}(1-\theta)}\\
    =&\theta.
    \end{split}
\end{equation}

However, when we take the finite sum of (\ref{rt:fai_0}) to approximate $\varphi_0$, it needs more items for larger $\theta$ than smaller ones, as shown in Fig. \ref{fig_4}.
\begin{figure}
\centering
\includegraphics[width=3.3in]{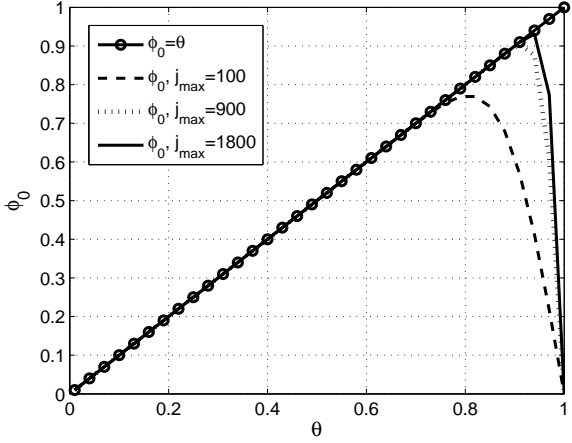}
\caption{The approximation of $\varphi_0$ by finite sums}\label{fig_4}
\end{figure}

On the other hand, $\Pr\{L^+>K\}$ gives the probability that the queue length is over $K$ in the infinite-buffer model, which can be seen as a \textit{virtual overflow probability} of the infinite buffer model. Define the virtual overflow probability for a given $K$ as
\begin{equation}\label{df:poverf_infity}
\begin{split}
    P_{overflow}^\infty=&\Pr\{L^+>K\}=\sum_{j=K+1}^\infty \pi_j=\varphi_K\\
    =&(1-\theta) \sum_{j=1}^\infty\frac{1}{(K+j)!}(j\theta)^{K+j}e^{-j\theta}.
\end{split}
\end{equation}

Next, define $\vec{\pi}^K=\{\pi_0^K,\pi_1^K,\cdots,\pi_K^K\}$ as the stationary distribution of the queue length process for the finite-buffer model, which is given by the following theorem.

\begin{theorem}\label{th:pi_piK}
    The stationary distribution of the queue length process for the finite-buffer model is given by
    \begin{equation}
            \begin{aligned}
            &\pi_j^K=\frac{\varphi_{j-1}-\varphi_j}{1-\theta \varphi_K}, &j=0,1,\cdots,K-1,\\
            &\pi_K^K=\frac{\varphi_{K-1}-\theta\varphi_K}{1-\theta \varphi_K}, &j=K,
        \end{aligned}
    \end{equation}
    where $\varphi_{-1}=1$ and $\varphi_k$ is given by (\ref{rt:fai_0}) for $k\geq0$.
\end{theorem}

\begin{proof}
    Define a \textit{cycle} as the time interval between two consecutive service completions leaving an empty buffer behind. According to regenerative process theory, the statistical properties in one cycle is the same as that of the entire process when the process has been in stationary state.

    For the convenience of the proof, let's give some definitions for the finite-buffer queue model, which are well defined under the assumption $\theta<1$.

    \begin{description}
      \item[$N^K$] The number of blocks in one cycle.
      \item[$N_j^K$] The number of blocks in one cycle at whose beginning there are $j$ packets in the buffer $(j=0,1,\cdots,K)$.
      \item[$S^K$] The number of packets served in one cycle.
      \item[$S_j^K$] The number of packets completions that there are $j$ packets left in the buffer at its departure $(j=0,1,\cdots,K-1)$. It is clear that $S_0^K=1$.
      \item[$\mu_j^K$] The expected number of blocks until the next packets completion when a packet has been just completed with $j$ packets left behind.
    \end{description}

    The corresponding quantities for the infinite-buffer model are denote by $N$, $N_j$, $S$, $S_j$ and $\mu_j$, respectively.

    By the theory of regenerative process, we have $\pi_j^K=\frac{\textsf{E}[N_j^K]}{\textsf{E}[N^K]}$ and $\pi_j=\frac{\textsf{E}[N_j]}{\textsf{E}[N]}$.

    Similarly, the transitions among states $\{0,1,\cdots,K-1\}$ are the same for both finite and infinite-buffer models. Besides, the service of each block is independent from each other. Thus the transitions among states above level $K$ do not affect the number of downcrossing to any state $j\leq K-1$ during one cycle in the infinite-buffer model. Note that this is not valid for $j=K$. Therefore,
    \begin{equation}\label{eq:N_j}
        \textsf{E}[N_j^K]=\textsf{E}[N_j],~j=0,1,\cdots,K-1.
    \end{equation}

    Let $\eta=\frac{\textsf{E}[N]}{\textsf{E}[N^K]}$, then we have
    \begin{equation}\label{eq:pipiK}
        \pi_j^K=\eta \pi_j,~~~j=0,1,\cdots,K-1.
    \end{equation}

    On the other hand, we have $\mu_j^K=\mu_j$ for $0\leq j\leq K-1$ and $\mu_j=\theta$ for $k\geq1$. Particularly, for $j=0$, there is some idle time besides the service time. We also have $\textsf{E}[S_j^K]=\textsf{E}[S_j]$ for $0\leq j\leq K-1$. By the definition of $N^K$ and $N$, their expected value are given by
    \begin{equation}
        \begin{split}
            \textsf{E}[N^K]=&\mu_0^K+\sum_{j=1}^{K-1}\textsf{E}[S_j^K]\mu_j^K,\\
            \textsf{E}[N]=&\mu_0+\sum_{j=1}^\infty\textsf{E}[S_j]\mu_j.
        \end{split}
    \end{equation}

    Using the equations above, it follows that
    \begin{equation}\label{eq:EN_ENK}
        \textsf{E}[N]=\textsf{E}[N^K]+\sum_{j=K}^\infty\textsf{E}[S_j]\theta.
    \end{equation}

    According to the definition of $\eta$, we have $\textsf{E}[N]/\eta=\textsf{E}[N^K]$ holds. Dividing (\ref{eq:EN_ENK}) by $\textsf{E}[N]/\eta$ and $\textsf{E}[N^K]$ on both sides, respectively, we have
    \begin{equation}\label{dr:eta_final}
        \eta=1 + \eta\theta\sum_{j=K}^\infty\frac{\textsf{E}[S_j]}{\textsf{E}[N]}.
    \end{equation}

    By the theory of regenerative process, the statistical property of one cycle converges to that of the whole process and  we have $\pi_j=\frac{\textsf{E}[S_j]}{\textsf{E}[S]}$. Since there is only one packet arrives in each block, it is clear that the length of a cycle equals to the number of packets served during the cycle. So we have $\textsf{E}[N]=\textsf{E}[S]$. Combining these discussions, we know that
    \begin{equation}
        \pi_j=\frac{\textsf{E}[S_j]} {\textsf{E}[N]}.
    \end{equation}
    With this equation and solving $\eta$ from (\ref{dr:eta_final}) one can get
    \begin{equation}\label{rt:eta}
        \eta=\frac{1}{1-\theta\sum_{j=K}^\infty \pi_j}=\frac{1}{1-\theta \varphi_K}.
    \end{equation}

    Combining (\ref{rt:pij}), (\ref{df:poverf_infity}), (\ref{eq:pipiK}) and (\ref{rt:eta}), \textit{Theorem} \ref{th:pi_piK} is proved.

\end{proof}

\subsection{Packet Delay}
\begin{proposition}
    For the finite-buffer model and the FIFO discipline, the average packet delay (the period from the arrival of a packet to its departure) is given by
    \begin{equation}
        \textsf{E}[\widehat D]=\frac{1}{2}+ \frac{\sum_{j=0}^{K-1} \varphi_j - K\theta \varphi_K}{1-\theta \varphi_K}+\int_0^1 (x+1) e^{\frac{-\theta}{x}} dx,
    \end{equation}
    where $\varphi_{-1}=1$ and $\varphi_k$ is given by (\ref{rt:fai_0}) for $k\geq0$.
\end{proposition}

\begin{proof}
Since the stationary queue length distribution has been given in Theorem \ref{th:pi_piK}, the average queue length is given by
\begin{equation}
\begin{split}
    &\textsf{E}[\widehat{L}^+]=\sum_{j=0}^K j  \pi_j^K\\
    =&\sum_{j=1}^{K-1} \frac{j(\varphi_{j-1}-\varphi_j)}{1-\theta \varphi_K} + K \left( 1-\sum_{j=0}^{K-1} \frac{\varphi_{j-1}-\varphi_j}{1-\theta \varphi_K} \right)\\
    =&\frac{1}{1-\theta \varphi_K}\left(\sum_{j=0}^{K-1} \varphi_j -K\varphi_{K-1}\right) +K - \frac{K(1-\varphi_{K-1})}{1-\theta \varphi_K} \\
    =&\frac{\sum_{j=0}^{K-1} \varphi_j - K\theta \varphi_K}{1-\theta \varphi_K}
\end{split}
\end{equation}

Since there is only one packet arrives in each block, the average time that a packet stays in the system (including its service time and waiting time) equals to the average queue length by the Little's law,
\begin{equation}
    \textsf{E}[\widehat T_Q]=\textsf{E}[\widehat{L}^+].
\end{equation}

On the other hand, the vestige time of the finite-buffer model won't change much compared with the infinite-buffer model, especially for large $K$. Thus, the average vestige time is approximated by $\textsf{E}[\widehat V]=\textsf{E}[V]$, which is given by (\ref{rt:E_V}).

Finally, the total average delay in the finite-buffer model is
\begin{equation}
    \textsf{E}[\widehat D]=\textsf{E}[\widehat T_Q]+\textsf{E}[\widehat V]
\end{equation}
which completes the proof.
\end{proof}

\subsection{Overflow probability}
Define the overflow probability $P_{overflow}$ of a finite size buffer as the long-run fraction of packets that are rejected due to the finite capacity of the buffer. This is also the probability that the buffer is in the overflow state at packet arriving epochs. Following the queueing analysis method, we will investigate the overflow probability of a finite size buffer in this section. Motivated by the result in \cite{Tijms-Fubian}, the overflow probability is obtained based on the results on the infinite-buffer model, and is given by the following theorem.

\begin{theorem}\label{lem:poverf_fandinf}
    Assume the buffer size is $K$ (unit: $L_p$) and the overflow probability is $P_{overflow}$. For the queue model considered in this paper (one packet arrives each block), the overflow probability of the finite-buffer aided communication over block Rayleigh fading channel in the low SNR regime can be given by
    \begin{equation}\label{rt:po}
        P_{overflow}=\frac{(1-\theta) \varphi_K}{1-\theta \varphi_K},
    \end{equation}
    where $\theta=\frac{L_p}{\nu}=\frac{L_p}{WT_B\rho}$ and $\varphi_K$ is given in (\ref{rt:fai_0}).

\end{theorem}

\begin{proof}

    Under the assumption $\theta<1$, both the finite and infinite-buffer models are stable, i.e., the queue will not be always in the overflow state, nor the queue length will go to infinity. In other words, each packet can be served within one block in the average sense, which implies that the service time required by one packet equals to the service time offered by the channel in each block. As we know, the average service time required by each packet is $\theta$, which is less than 1. On the other hand, the channel is busy except a probability of $\pi_0$. So the service time provided by the channel in each block is $1-\pi_0$ in the average sense. Then for the infinite-buffer model, the following equation holds.
    \begin{equation}\label{eq:bta_piK}
        \theta=1-\pi_0.
    \end{equation}

    For the finite-buffer case, a fraction of $P_{overflow}$ packets are discarded. Thus the actual amount of service time required in each block becomes $(1-P_{overflow})\theta$. Therefore, we have
    \begin{equation}\label{eq:bta_pi}
        (1-P_{overflow})\theta=1-\pi_0^K.
    \end{equation}

    Combining (\ref{eq:pipiK}), (\ref{eq:bta_piK}) and (\ref{eq:bta_pi}) and solving $P_{overflow}$, we have
    \begin{equation}\label{dr:p_overf}
        P_{overflow}=\frac{(\eta-1)(1-\theta)}{\theta}.
    \end{equation}

    By replacing the $\eta$ with the result in (\ref{rt:eta}), we finally get the overflow probability as
    \begin{equation}\label{rt:p_overf}
        P_{overflow}=\frac{(1-\theta)\varphi_K}{1-\theta \varphi_K},
    \end{equation}
    which completes the proof.

\end{proof}
\begin{figure}[h]
\centering
\includegraphics[width=3.5in]{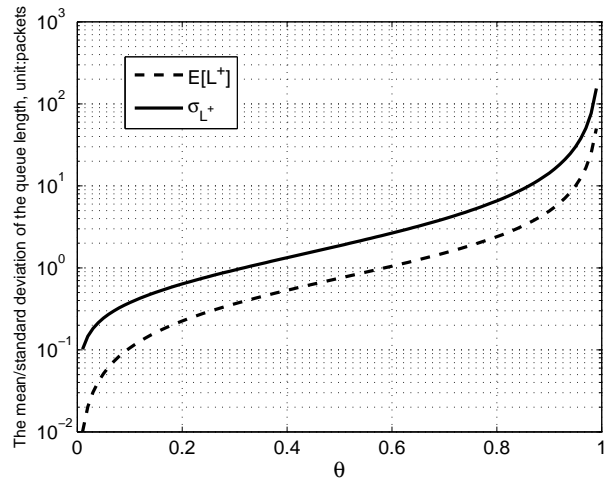}
\caption{The mean/standard deviation of the queue length in the infinite-buffer model} \label{fig:EL}
\end{figure}

\section{Numerical Results}\label{sec:6}
In this section, some numerical results are presented to illustrate the stationary distribution, packet delay, as well as the overflow probability for the finite-buffer model.
The component variance is assumed as to be $\sigma^2=1$, the system bandwidth is $W=5$ MHz and transmitting power is $P=-10$ dBW. Suppose that the distance between the transmitter and the receiver is $d=1000$ m and the pathloss exponent is $\alpha=4$. The block length is chosen as $T_B=10^{-4}$ s.

According to its definition, we have $\theta=\frac{L_p}{WT_B\rho}$.
Remember that $L_p=RT_B$ and the equivalent AWGN capacity of the channel is $C_a=W\rho$ in the low SNR regime. Therefore, $\theta$ can also be seen as the ratio between the traffic rate and the ergodic capacity.
\begin{figure}
\centering
\includegraphics[width=3.5in]{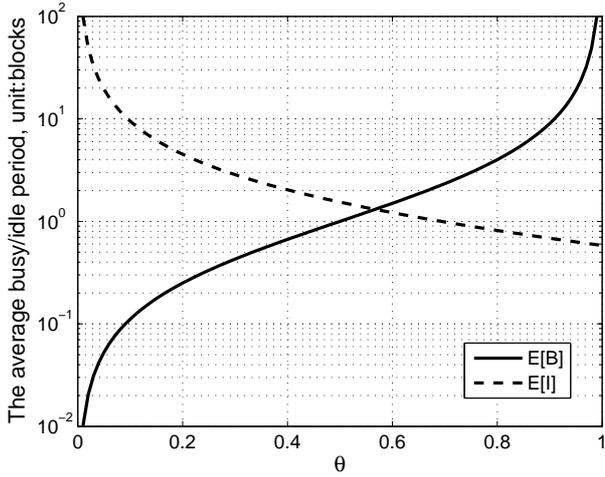}
\caption{The average busy/idle period of the infinite-buffer model} \label{fig:EB}
\end{figure}

\begin{figure}
\centering
\includegraphics[width=3.5in]{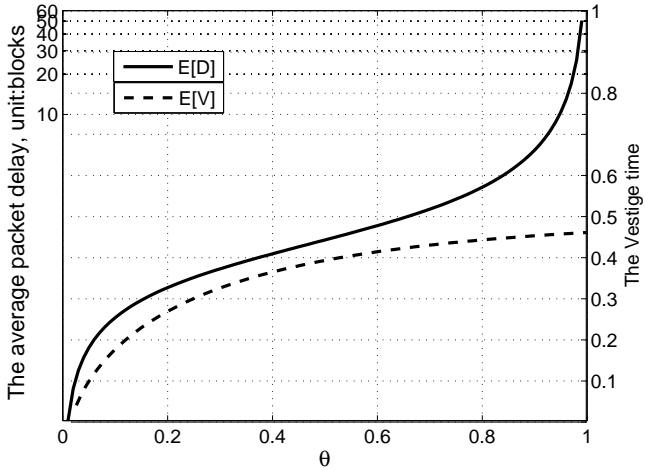}
\caption{The average packet delay of the infinite-buffer model. E[D] corresponds to the log-scale y-axis on the left while E[V] corresponds to the linear y-axis on the right.} \label{fig:ED}
\end{figure}

\begin{figure}
\centering
\includegraphics[width=3.5in]{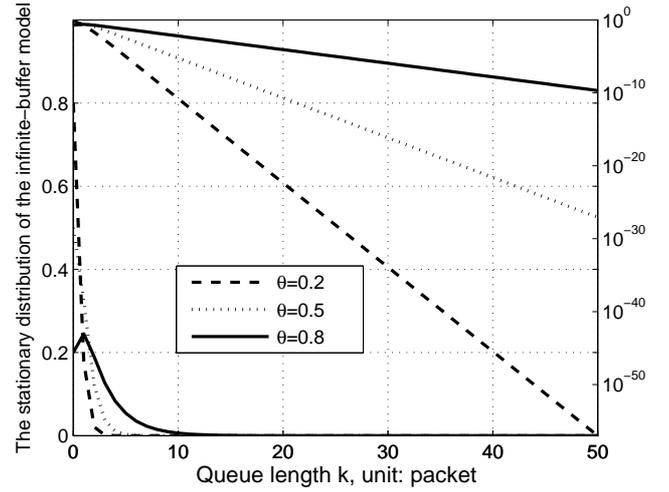}
\caption{The stationary queue length distribution of the infinite-buffer model. Each curve is presented both in linear y-axis on the left and log-scale y-axis on the right.} \label{fig:piinf}
\end{figure}

\begin{figure}
\centering
\includegraphics[width=3.5in]{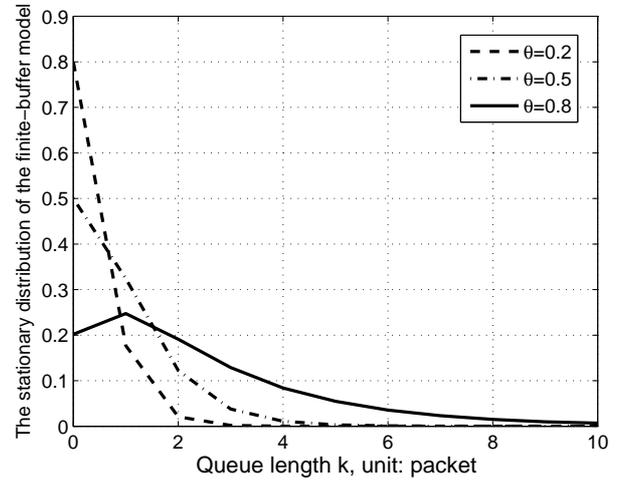}
\caption{The stationary queue length distribution of the finite-buffer model, buffer size $K=10$ packets} \label{fig:piK10}
\end{figure}

\begin{figure}
\centering
\includegraphics[width=3.5in]{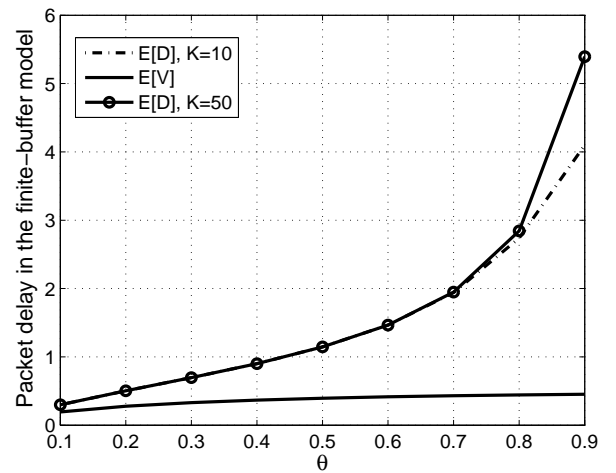}
\caption{The average packet delay (blocks) of the finite-buffer model} \label{fig:EDf}
\end{figure}

\begin{figure}
\centering
\includegraphics[width=3.5in]{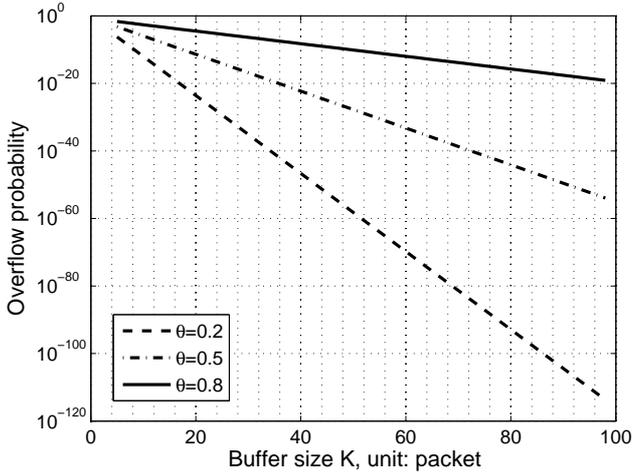}
\caption{The overflow probability versus buffer size (packets)} \label{fig:po}
\end{figure}

For the infinite-buffer model, the mean value as well as the standard deviation of the queue length $L^+$ are presented in Fig. \ref{fig:EL}. When $\theta$ is increased, $\textsf{E}[L^+]$ and $\sigma_L$ increases accordingly. Particularly, they will go to infinity when $\theta$ approaches 1. This means that, blocks fading channels can not support a data stream at a rate close to its egodic capacity. Its busy period and idle period are given by (\ref{rt:E_busy}) and shown in Fig. \ref{fig:EB}. It is seen that $\textsf{E}[B]$ increases while $\textsf{E}[I]$ decreases quickly with $\theta$, which means that the utilization of channel is improved. However, the average packet delay will increase quickly at the same time, as shown in Fig. \ref{fig:ED}. Particularly, the average delay corresponds to the log-scale y-axis on the left and the average vestige time corresponds to the linear y-axis on the right, both in blocks. It is seen that the average vestige time $\textsf{E}[V]$ is also increasing with $\theta$ but never exceeds 0.5 blocks.

The stationary queue length distribution of the infinite-buffer model is shown in Fig. \ref{fig:piinf} for $\theta=0.2,0.5,0.8$, where they are presented in both linear and log-scale y-axis. It is clear that the larger $\theta$ is, the slower $\pi_k$ decreases. This is easy to understand because heavier traffic leads to longer queues. Most importantly, as the queue length  grows, the probability that it occurs decreases approximately exponentially.

When the buffer size is finite, the stationary queue length distribution is slightly changed due to buffer overflows. Assume buffer size is $K=10$ packets, the queue length distribution $\vec\pi^K$ is shown in Fig. \ref{fig:piK10}. It is seen that the queue length is generally small and is increasing with $\theta$. Similarly, the packet delay grows with both $\theta$ and buffer size, as shown in Fig. \ref{fig:EDf}.

Finally, the overflow probability of a finite buffer is given in Fig. \ref{fig:po}. For a smaller $\theta$, i.e., slight traffics, the overflow probability is smaller and decreases more quickly. However, as is shown, all the overflow probability decrease exponentially in all cases.

\newpage

\section{Conclusion}\label{sec:7}

Modern communications requires the wireless channels to provides QoS (quality of service) guaranteed services. How to characterize and make full use of the service capability of fading channels is an urgent problem. In this paper, we solved the problem in the case of low SNR block Rayleigh fading channels by taking advantage of the memoryless property of their one block services. However, for general fading channels, the answers are quite unclear, which needs a lot of further efforts.

\section*{Acknowledgement}
This work was partly supported by the China Major
State Basic Research Development Program (973 Program) No.
2012CB316100(2), National Natural Science Foundation of China
(NSFC) No. 61171064 and NSFC No. 61021001.

\appendices
\renewcommand{\theequation}{\thesection.\arabic{equation}}

\newcounter{mytempthcnt}
\setcounter{mytempthcnt}{\value{theorem}}
\setcounter{theorem}{2}

\section{Proof of Theorem 3}\label{prf:them3}

\begin{theorem}\label{th:delay}
    In the FIFO discipline, the average packet delay is given by
        \begin{equation}\label{rt:E_D}
            \textsf{E}[D]=\frac{1}{2}+\theta+\frac{\theta^2}{2(1-\theta)}+\int_0^1 (x-1) e^{\frac{-\theta}{x}} dx.
        \end{equation}
\end{theorem}

\begin{proof} \quad

    \textit{A. The Service and Waiting Time}

    Since there is only one packet arrives in each block, the average time that a packet stays in the system ($S+W$) equals to the average queue length by the Little's law. So we have
    \begin{equation}\label{dr:sw}
        \textsf{E}[S+W]=\textsf{E}[L^+]=\lim_{z\rightarrow1}L^+(z)=\frac{\theta(2-\theta)}{2(1-\theta)}.
    \end{equation}

    \textit{B. The Vestige Time}

    To investigate the vestige time $V$, we have to define some new random variables. Firstly, for $k\geq1$, define
    \begin{equation}
        \begin{split}
            &U_k=L_p-S_k,~~~U_k^+=U_k|_{U_k>0},\\
            &V_k=\frac{U_k^+}{s_n},~~~~~~~~V_k^-=V_k|_{V_k<1},
        \end{split}
    \end{equation}
    where $s_n$ is the service provided by the channel in one block and $S_k$ is the total amount of service provided by $k$ successive blocks defined in (\ref{df:Sk}).
    By the definition, $U_k$ is the remaining part of a packet after $k$ blocks of transmission, which is not necessarily positive. Using a positive condition, we get $U_k^+$. In this case, there is really part of the packet untransmitted. Assuming the channel service of the next block is $s_n$, the remaining amount of data $U_k^+$ needs $V_k$ blocks to be transmitted. However, $V_k$ is the vestige time only if the remaining data can be transmitted within this block, which reduces to $V_k^+$.

    As one can see, $U_k\in(-\infty,L_p)$, $U_k^+\in(0,L_p)$, $V_k\in(0,\infty)$ and $V_k^+\in(0,1)$. Particularly, the CDF of $U_k$ can be derived as follows.
    \begin{equation}
        \begin{split}
        F_{U_k}(x)=&\Pr\{U_k\leq x\}=\Pr\{S_k\geq L_p-x\}\\
        =&1-F_{S}(L_p-x)=\Gamma(k,\frac{L_p-x}{\nu})\\
        =&\frac{1}{\Gamma(k)}\int_{\frac{L_p-x}{\nu}}^\infty t^{k-1}e^{-t} dt,
        \end{split}
    \end{equation}
    where $\Gamma(k,x)=\frac{1}{\Gamma(k)}\int_x^\infty t^{k-1}e^{-t} dt$ is the upper Gamma function and $\nu=WT_B\rho$.

    Therefore, the CDF of $U^+$ is
    \begin{equation}
        \begin{split}
            F_{U_k^+}=&\Pr\{U_k^+\leq x|U_k^+>0\}
            =\frac{\Pr\{0<U_k^+\leq x\}}{\Pr\{U_k^+>0\}}\\
            =&\frac{F_{U_k}(x)-F_{U_k}(0)}{1-F_{U_k}(0)}
            =\frac{\int_{\frac{L_p-x}{\nu}}^\infty t^{k-1} e^{-t} dt - \Gamma(k, \theta)}   {\gamma(k,\theta )},
        \end{split}
    \end{equation}
    where $\gamma(k,x)=\frac{1}{\Gamma(k)}\int_0^x t^{k-1}e^{-t} dt$ is the lower Gamma function. For $x\in(0,L_p)$, taking derivative on $F_{U_k^+}$, we get the {\em p.d.f.} of $U_k^+$ as
    \begin{equation}
        f_{U_k^+}(x)=\frac{1}{\nu^k\gamma(k,\theta)} (L_p-x)^{k-1} e^{-\frac{L_p-x}{\nu}}.
    \end{equation}

    For some $x\in(0,+\infty)$, define $D$ as the area in the $U_k^+-s_n$ plane that $0<U_k^+<L_p$ and $s_n>\frac{U_k^+}{x}$. Since $U_k^+$ and the channel service of the next block ($s_n$) are random variables independent from each other, the CDF of $V_k$ can be derived as follows.
    \begin{equation}
        \begin{split}
        &F_{V_k}(x)=\Pr\{V_k\leq x \}=\Pr\{U_k^+\leq xs_n \}\\
        =&\iint_D f_{U_k^+}(u)f_s(s)duds\\\\
        =&\int_0^{L_p} f_{U_k^+}(u) du \int_{\frac{u}{x}}^\infty \frac{1}{\nu} e^{-\frac{s}{\nu}} ds\\
        =&\frac{1}{\nu^k\gamma(k,\theta)} \int_0^{L_p} (L_p-u)^{k-1}e^{-\frac{1}{\nu}{(L_p-u+\frac{u}{x})}} du.
        \end{split}
    \end{equation}

    So one can get the probability that $V_k$ is less than 1 as
     \begin{equation}
        \begin{split}
        \Pr\{V_k<1\}=&\frac{1}{\nu^k\gamma(k,\theta)} \int_0^{L_p} (L_p-u)^{k-1}e^{-\frac{L_p}{\nu}} du\\
        =&\frac{e^{-\theta}}{\nu^k\gamma(k,\theta)} \frac{L_p^k}{k}
        \end{split}
    \end{equation}

    Thus, the CDF of $V_k^-$ is
    \begin{equation}
        \begin{split}
            F_{V_k^-}(x)=&\Pr\{V_k^-\leq x|V_k<1\}=\frac{\Pr\{V_k<x\}}{\Pr\{V_k<1\}}\\
            =&e^{\theta} \frac{k}{L_p^k} \int_0^{L_p} (L_p-u)^{k-1}e^{-\frac{1}{\nu}{(L_p-u-\frac{u}{x})}} du\\
            =&\frac{k}{L_p}\int_0^{L_p} \left(\frac{L_p-u}{L_p}\right)^{k-1}e^{\frac{u}{\nu}{(1-\frac{1}{x})}} du,
        \end{split}
    \end{equation}
    for $x\in(0,1)$ and the average of $V_k^-$ can be expressed as
    \begin{equation}\label{dr:E_xk_}
        \begin{split}
            \textsf{E}[V_k^-]=&\int_0^1 x dF_{V_k^-}(x)
            \stackrel{(a)}{=}xF_{V_k^-}(x)|_0^1-\int_0^1 F_{V_k^-}(x) dx\\
            =&1-\frac{k}{L_p}\int_0^1 \int_0^{L_p} \left(\frac{L_p-u}{L_p}\right)^{k-1}e^{\frac{u}{\nu}{(1-\frac{1}{x})}} dxdu,
        \end{split}
    \end{equation}
    where (a) follows the integration by parts.

    For the case of $k=0$, define $V_0=\frac{L_p}{s_n}$ and $V_0^-=V_0|_{V_0<1}$. Specifically, if one packet is completed within one block, the vestige time equals to its actual service time. The CDF of $V_0$ and $V_0^-$ are, respectively
    \begin{equation}
        \begin{split}
            F_{V_0}=&\Pr\{V_0\leq x\}=\Pr\{s_n\geq \frac{L_p}{x}\}
            =e^{-\frac{\theta}{x}}
        \end{split}
    \end{equation}
    for $x\in(0,+\infty)$   and
    \begin{equation}
        \begin{split}
            F_{V_0^-}=&\Pr\{V_0\leq x|V_0<1\}=\frac{\Pr\{V_0\leq x\}} {\Pr\{V_0\leq 1\}}
            =e^{\theta\left(1-\frac{1}{x}\right)}
        \end{split}
    \end{equation}
    for $x\in(0,1)$.

    Then we can get the expected value of $V_0^-$ as
    \begin{equation}\label{dr:E_V0}
        \begin{split}
            \textsf{E}[V_0^-]=&\int_0^1 xdF_{V_0^-}(x)\\
            =&xe^{\theta\left(1-\frac{1}{x}\right)}|_0^1 - \int_0^1 e^{\theta\left(1-\frac{1}{x}\right)}dx\\
            =&1-e^\theta\int_0^1 e^{-\frac{\theta}{x}}dx.
        \end{split}
    \end{equation}

    Finally, by the whole probability formula, the expected value of vestige time $V$ is
    \begin{equation}\label{dr:E_V}
           \textsf{E}[V]=\textsf{E}[V_0^-]\Pr\{T=0\}+\sum_{k=1}^\infty \textsf{E}[V_k^-]\Pr\{T=k\},
    \end{equation}
    where
    \begin{equation}\label{dr:E_Vremaining}
        \begin{split}
            &\sum_{k=1}^\infty \textsf{E}[V_k^-]\Pr\{T_n=k\}\\
            =&\sum_{k=1}^\infty \left[ 1-\frac{k}{L_p}\int_0^1 \int_0^{L_p} \left(\frac{L_p-u}{L_p}\right)^{k-1}e^{\frac{u}{\nu}{(1-\frac{1}{x})}} dxdu \right]\frac{1}{k!}e^{-\theta} \theta^k\\
            =&1-e^{-\theta}- \frac{k}{L_p} \int_0^1\int_0^{L_p} \sum_{k=1}^\infty \left(\theta \frac{L_p-u}{L_p}\right)^{k-1} \frac{\theta}{k!}e^{-\theta} e^{\frac{u}{\nu}{(1-\frac{1}{x})}} dxdu\\
            \stackrel{(a)}{=}&1-e^{-\theta}- \frac{\theta e^{-\theta}}{L_p} \int_0^1\int_0^{L_p} e^{ \theta\left(1- \frac{u}{L_p}\right) } e^{\frac{u}{\nu}{(1-\frac{1}{x})}} dxdu\\
            \stackrel{(b)}{=}&1-e^{-\theta}- \frac{\theta}{L_p} \int_0^1\int_0^{L_p}  e^{\frac{-u}{\nu x}} dxdu\\
            =&1-e^{-\theta}- \int_0^1 x \left(1-  e^{\frac{-\theta}{x}}\right) dx\\
            =&\frac{1}{2}-e^{-\theta}+\int_0^1 x e^{\frac{-\theta}{x}} dx,
        \end{split}
    \end{equation}
    where (a) follows $\sum_{k=1}^\infty \frac{1}{(k-1)!}x^{k-1}=e^x$ and (b) uses the equation $\frac{\theta}{L_p}\cdot \nu=1$.

    Combining (\ref{dr:E_V0}), (\ref{dr:E_V}) and (\ref{dr:E_Vremaining}) we have
    \begin{equation}\label{rt:E_V}
        \begin{split}
            \textsf{E}[V]=&e^{-\theta}\cdot\left(1-e^\theta\int_0^1e^{-\frac{\theta}{x}}dx\right)+\frac{1}{2}-e^{-\theta}+\int_0^1 x e^{\frac{-\theta}{x}} dx\\
            =&\frac{1}{2}+\int_0^1 (x-1) e^{\frac{-\theta}{x}} dx.
        \end{split}
    \end{equation}

    Finally, combing (\ref{df:delay}), (\ref{dr:sw}) and (\ref{rt:E_V}), we have
    \begin{equation}\label{rt:E_D_1}
        \textsf{E}[D]=\theta+\frac{\theta^2}{2(1-\theta)}+\frac{1}{2}+\int_0^1 (x-1) e^{\frac{-\theta}{x}} dx,
    \end{equation}
    which completes the proof of Theorem \ref{th:delay}.

\end{proof}

\newcounter{mytemplemcnt}
\setcounter{mytemplemcnt}{\value{lemma}}
\setcounter{lemma}{0}

\bibliographystyle{IEEEtran}

\end{document}